\documentclass[11pt]{article}
\usepackage{fullpage,amsmath,amssymb,amsthm}

\newcommand{\R}{\mathbb{R}}

\newcommand{\E}{\mathbb{E}}

\newcommand{\sk}{\mathsf{sk}}

\newcommand{\eps}{\epsilon}

\newtheorem{definition}{Definition}[section]

\newtheorem{lemma}[definition]{Lemma}
\newtheorem{theorem}[definition]{Theorem}

\theoremstyle{definition}
\newtheorem{remark}{Remark}
\newcommand{\plog}{\mathrm{polylog}}

\title{Optimal Lower Bounds for Sketching Graph Cuts}

\author{Charles Carlson\footnote{{\sf ccarlsn2@illinois.edu.} U.I.U.C.}\and Alexandra Kolla \footnote{{\sf akolla@illinois.edu.} U.I.U.C. Supported by NSF Grant 1423452. }\and Nikhil Srivastava\footnote{{\sf nikhil@math.berkeley.edu.} {U.C.
Berkeley}.  Supported by NSF grant CCF-1553751 and a Sloan research fellowship.} \and Luca Trevisan\footnote{{\sf luca@berkeley.edu.} U.C. Berkeley. Supported by the National Science Foundation under Grants No. 1540685 and No. 1655215.}}

\begin{document}

\maketitle
\begin{abstract} We study the space complexity of sketching cuts and Laplacian quadratic forms of graphs. 
	We show that any data structure which approximately stores the sizes of all cuts in an undirected
	graph on $n$ vertices up to a $1+\epsilon$ error must use $\Omega(n\log n/\epsilon^2)$ bits of space in the worst case,
	improving the $\Omega(n/\epsilon^2)$ bound of \cite{andoni} and matching the best known upper bound achieved by spectral sparsifiers \cite{bss}.
	Our proof is based on a rigidity phenomenon for cut (and spectral) approximation which may be of independent interest: 
	any two $d-$regular graphs which approximate each other's cuts significantly better than a random graph approximates
	the complete graph must overlap in a constant fraction of their edges. 
\end{abstract}
\section{Introduction}
An $\epsilon-$spectral sparsifier \cite{st04} of a weighted graph $G$ is a
weighted graph $H$ such that for every $x\in\R^n$  we have the multiplicative
guarantee: \begin{equation}\label{eqn:specapprox} (1-\epsilon) x^TL_Gx \le
x^TL_Hx \le (1+\epsilon)x^TL_Gx.\end{equation} An $\epsilon-$cut sparsifier
\cite{bk96} is a graph such that \eqref{eqn:specapprox} holds for all
$x\in\{-1,1\}^n$; the former condition clearly implies the latter, and it is
not hard to see that the converse is not true.  

There have been many papers on efficiently constructing sparsifiers with few
edges (e.g. \cite{ss,fhhp,zhu,lee2017sdp}). The best known
sparsity-approximation tradeoff is achieved in \cite{bss}, who showed that
every graph on $n$ vertices has an $\epsilon$-spectral sparsifier with
$O(n/\epsilon^2)$ edges (and this is also the best known upper bound for
$\epsilon$-cut sparsification). Since every weighted edge can be stored using
$O(\log n)$ bits\footnote{After discretizing the weights to, say, $1/n^4$
precision, which introduces negligible error.}, storing such a sparsifier
requires $O(n\log n/\epsilon^2)$ bits.

The recent work \cite{andoni} studied the question of whether it is possible to
use substantially fewer bits if one simply wants a data structure (not
necessarily a graph) which answers cut and quadratic form queries.  We will
use the following definition, which is inspired by the one in \cite{sidford}.

\begin{definition} An {\em $\epsilon-$spectral sketch} of a graph $G$ is 
	a function $f:\R^n\rightarrow\R$  such that
	for every $x\in\R^n$:
	\begin{equation}\label{eqn:sketchapprox}(1-\epsilon)x^TL_Gx \le f(x)\le (1+\epsilon)x^TL_Gx.\end{equation}
		An {\em $\epsilon-$cut sketch} is a function $f:\{-1,1\}^n\rightarrow\R$ such that \eqref{eqn:sketchapprox} holds for all $x\in\{-1,1\}^n$.

	A {\em sketching scheme} is a deterministic map $\sk$ from graphs on
	$n$ vertices to ($\eps$-spectral or $\eps$-cut) sketches
	$f:\R^n\rightarrow\R$, along with specified procedures for storing the
	functions as bit strings and for evaluating any given $f$ on a query
	$x$ \footnote{We will not be concerned with the details of these procedures since we are only interested in the space used by the sketches and
	not computational parameters such as query time / success probability
	of randomized schemes.}.
	The number of bits required to store $f$ is called its {\em size}.
\end{definition}
In \cite{andoni} it was shown that any $\eps$-cut sketch must use
$\Omega(n/\epsilon^2)$ bits in the worst case, leaving a logarithmic gap
between the best known upper and lower bounds. In this paper, we close this gap by 
showing that any $\epsilon$-cut sketching scheme must in fact use $\Omega(n\log n/\epsilon^2)$ bits whenever $\epsilon=\omega(n^{-1/4})$ (Theorem \ref{thm:cutlb}),
which means that it is not in general possible to obtain any asymptotic savings by considering sketches
which are not graphs. We also give a lowerbound for $\epsilon$-spectral sketching with a simpler proof and 
slightly better constants (Theorem \ref{thm:speclb}).

\subsection{Related Work}
The paper \cite{andoni} also studied the problem of producing 
a ``for each'' sketching algorithm, which has the property that for any particular query $x$,
$f(x)$ approximates $x^TLx$ with constant probability. They showed that in this weaker model 
it is possible to obtain $\eps-$cut sketches of size $O(n\plog (n) /\epsilon)$ and $\eps$-spectral sketches of size $O(n\plog (n)/\epsilon^{1.6})$.
The latter bound was recently improved to $O(n\plog (n)/\epsilon)$ and generalized to include queries to the pseudoinverse
in \cite{sidford}. 

In contrast, in this paper we study only the ``for all'' model, in which the sketch is required to work for all queries simultaneously.

\subsection{Techniques} Our proof is based on the following rigidity
phenomenon: there is a constant $c$ such that if two $d-$regular graphs
$c/\sqrt{d}$-cut approximate each other then they must overlap in a constant
fraction of edges (Lemmas \ref{lem:specrigid} and \ref{lem:cutrigid}). This
allows us to argue that below this approximation threshold any sketch encodes a
constant fraction of the information in the graph, and so any sketch must use
$\Omega(dn\log n)=\Omega(n\log n/\epsilon^2)$ bits.

Interestingly, this phenomenon is very sensitive to the value of the constant
$c$ --- in particular, a well-known theorem of Friedman \cite{friedman} says
that the eigenvalues of the Laplacian of a random $d-$regular graph are
contained in the range $[d-2\sqrt{d-1}-o(1),d+2\sqrt{d}+o(1)]$ with high
probability, which implies that almost every $d-$regular graph has a
$3/\sqrt{d}$-spectral sketch of constant size, namely the complete graph. 

Note that according to the Alon-Boppana bound \cite{alon} it is not possible to approximate
$K_n$ by a $d-$regular graph with error less than $2/\sqrt{d}$.  Thus, our
result can be interpreted as saying that a $d-$regular graph can only be
approximated by nearby graphs when the error is substantially below the Alon-Boppana bound.

We remark that the proof of \cite{andoni} was based on showing a $\Omega(1/\epsilon^2)$ bit lower bound
for a constant sized graph by reducing to the Gap-Hamming problem in communication complexity, and then
concatenating $O(n)$ instances of this problem to obtain an instance of size $n$. This method is thereby
inherently incapable of recovering the logarithmic factor that we obtain (in contrast, we work with random regular graphs).

\subsection{Notation and Organization} We will denote $\epsilon$-spectral approximation as $$(1-\epsilon)L_G\preceq L_H\preceq (1+\epsilon)L_G$$ and
$\epsilon-$cut approximation as $$(1-\epsilon)L_G\preceq_\square L_H\preceq_\square (1+\epsilon)L_G.$$ We will use $\lg$ and $\ln$ to
denote binary and natural logarithms.

We prove our lower bound for spectral sketching in Section 2, and the lower bound for cut sketching in Section 3.
Although the latter is logically stronger than the former, we have chosen to present the spectral case
first because it is extremely simple and introduces the conceptual ideas of the proof.

\subsection*{Acknowledgments} We would like to thank MSRI and the Simons Institute for the Theory of Computing, where this work was carried out
during the ``Bridging Discrete and Continuous Optimization'' program.

\section{Lower Bound for Spectral Sketches}
In this section, we prove the following theorem. 
\begin{theorem}\label{thm:speclb} For any $\epsilon=\omega(n^{-1/4})$, any $\epsilon$-spectral sketching scheme $\sk$ for graphs with $n$ vertices
must use at least $\frac{n\lg n}{500\eps^2}\cdot (1-o_n(1))$ bits in the worst case.\end{theorem}
The main ingredient is the following lemma.
\begin{lemma}[Rigidity of Spectral Approximation]\label{lem:specrigid}
Suppose $G$ and $H$ are simple $d-$regular graphs such that
	$$ (1-\eps)L_G\preceq L_H\preceq (1+\eps)L_G.$$
	Then $G$ and $H$ must have at least $\frac{dn}{2}(1-\epsilon^2d/2)$ edges in common.
	\end{lemma}
\begin{proof}
Since $G$ and $H$ are $d-$regular, we have $L_H-L_G =
	(dI-A_H)-(dI-A_G)=A_G-A_H$ and $\|L_G\|\le 2d$. Thus the hypothesis
	implies that:		\[ -2d\epsilon I\preceq -\epsilon\cdot
	L_G\preceq L_H-L_G = A_G-A_H\preceq \epsilon\cdot L_G\preceq 2d\epsilon
	I,
\]
which means $\|A_G-A_H\|\le 2\epsilon d$. Passing to the Frobenius norm, we find that
	$$ \sum_{ij}(A_G(i,j)-A_H(i,j))^2=\|A_G-A_H\|_F^2\le n\|A_G-A_H\|^2\le 4\epsilon^2d^2 n.$$
	The matrix $A_G-A_H$ has entries in $\{-1,0,1\}$, with exactly
	two nonzero entries for every edge in $E(G)\Delta E(H)$, so the left hand side is equal to 
	$$2|E(G)\Delta E(H)| = 2(|E(G)|+|E(H)|-2|E(G)\cap E(H)|) = 2(dn-2|E(G)\cap E(H)|),$$
	since $|E(G)|=|E(H)|=dn/2$. Rearranging yields the desired claim.
\end{proof}
Note that the above lemma is vacuous for $\epsilon \ge\sqrt{2/d}$ but indicates
that $G$ and $H$ must share a fraction of their edges for $\epsilon$ below this
threshold.
\begin{lemma}\label{lem:sketchcount} For any function $f:\R^n\rightarrow \R$ and $\epsilon<1/2$:
	$$\lg |\{G\in G_{n,d}:\textrm{$G$ is $\epsilon-$spectrally approximated by $f$}\}|\le dn/2+5\epsilon^2d^2n\lg n,$$
	where $G_{n,d}$ denotes the set of all $d-$regular graphs on $n$ vertices.
\end{lemma}
\begin{proof} If $f$ does not $\epsilon-$approximate any $d-$regular graph then we are done. Otherwise, let $H$ be the 
	lexicographically (in some pre-determined ordering on $d-$regular graphs) first graph which $f$ $\epsilon$-approximates.
	Suppose $G$ is another graph that $f$ $\epsilon-$approximates. Notice that by applying \eqref{eqn:sketchapprox} twice, we have
	that for every $x\in\R^n$:
	$$ \frac{1-\epsilon}{1+\epsilon}\le\frac{x^TL_Hx}{x^TL_Gx} = \frac{f(x)}{x^TL_Gx}\cdot \frac{x^TL_Hx}{f(x)}\le\frac{1+\epsilon}{1-\epsilon},$$
	so $G$ $3\epsilon$-spectrally approximates $H$. By Lemma \ref{lem:specrigid}, $H$ and $G$ must share 
	$$ k:=\frac{dn}{2}(1-9\epsilon^2d/2)$$
	edges. Thus, $G$ can be encoded by specifying:
	\begin{enumerate}
	\item Which edges of $H$ occur in $G$. This is a subset of the edges of $H$, which requires at most $dn/2$ bits.
	\item The remaining $dn/2-k = 9\epsilon^2d^2n/4$ edges of $G$. Each edge requires at most $2\lg(n)$ bits to specify,
		so the number of bits needed is at most $18\epsilon^2d^2n\lg n/4$.
	\end{enumerate}
	Thus, the total number of bits required is at most
	$$ {dn/2+18\epsilon^2d^2n\lg n/4},$$
	as desired.
\end{proof}

\begin{theorem}{(\cite{wormald}{Cor. 2.4})}\label{thm:mckay} For $d=o(\sqrt{n})$, the number of $d-$regular graphs on $n$ vertices is:
	$$ \frac{(dn)!}{(dn/2)!2^{dn/2}(d!)^n}\exp\left(\frac{1-d^2}{4}-\frac{d^3}{12n}+O\left(\frac{d^2}{n}\right)\right)
	\\\ge \exp(dn\ln(n/d)/2\cdot(1-o_n(1))).$$
\end{theorem}

\begin{proof}[Proof of Theorem \ref{thm:speclb}]  Let $N$ be the number of distinct sketches produced by
		$\sk$ and let $d=\lceil \frac{1}{25\eps^2}\rceil$. By Lemma
		\ref{lem:sketchcount}, the binary logarithm of the number of
		$d-$regular graphs $\epsilon$-spectrally approximated by any
		single sketch $f\in\mathrm{range}(\sk)$ is at most $${dn/2+dn\lg n/5}\le dn\lg n/5\cdot (1+o(1)).$$
		On the other hand by Theorem \ref{thm:mckay}, since $d=o(\sqrt{n})$ we have
		\begin{equation}\label{eqn:manygraphs}\lg|G_{n,d}|\ge dn\lg(n/d)/2\cdot(1-o(1))\ge dn\lg n/4\cdot (1-o(1)). \end{equation}
		Since every $d-$regular graph receives a sketch, we must have
		$$(1+o(1))\lg N \ge dn\lg n/4-dn\lg n/5 = dn\lg n/20=n\lg n/500\epsilon^2,$$
		as desired.
	
	\end{proof}
\begin{remark} The proof above actually shows that any $1/5\sqrt{d}$-spectral sketching scheme for $d-$regular graphs on $n$ 
	vertices must use at least $\Omega(dn\lg n)$ bits {\em on average}, since the same proof goes through if we only 
	insist that the sketches work for {\em most} graphs. \end{remark}

\begin{remark} The result of \cite{bss} produces $\epsilon$-spectral sparsifiers with $16n/\epsilon^2$ edges,
which yield $\epsilon$-spectral sketches with $64n\lg n/\epsilon^2$ bits by discretizing the edge weights up to $1/n^2$ error,
so the bound above is tight up to a factor of $64\cdot 500$. We have not made any attempt to optimize the constants. 
\end{remark}
\section{Lower Bound for Cut Sketches}
In this section we prove Theorem \ref{thm:cutlb}. The new ingredient is a rigidity lemma for cuts, which may be seen as a discrete analogue
of Lemma \ref{lem:specrigid}. The lemma holds for bipartite graphs and is proven using a Goemans-Williamson \cite{goemans} style rounding argument.
\begin{lemma}[Rigidity of Cut Approximation]\label{lem:cutrigid}
	Suppose $G$ and $H$ are simple $d-$regular bipartite graphs with the same bipartition $L\cup R$, such that
	\begin{equation}\label{eqn:cutappx} (1-\eps)L_G\preceq_\square L_H\preceq_\square (1+\eps)L_G.\end{equation}
		Then $G$ and $H$ must have at least $\frac{dn}{2}(1-3\sqrt{d}\epsilon)$ edges in common.
	\end{lemma}
	\begin{proof}
We will show the contrapositive. Assume $$|E(G)\setminus
		E(H)|=|E(H)\setminus E(G)|=\delta dn/2$$ for some $\delta\ge 3\sqrt{d}\epsilon$.  To show that
\eqref{eqn:cutappx} does not hold, it is sufficient to exhibit a vector $x\in\{-1,1\}^n$ 
such that 
\[ x^TMx:=x^T(L_G - L_H)x > 2\epsilon dn \ge \epsilon \cdot x^T L_G x,  \]
where $M=L_G-L_H=A_H-A_G$, since both graphs are $d-$regular and the latter inequality follows from $\|L_G\|\le 2d$.
To find such an $x$ we will first construct $n$ vectors $y_1,\ldots,y_n\in\R^n$ such that
		\begin{equation}\label{eqn:yiyj}
\sum_{i,j = 1}^n M_{ij} \langle y_i, y_j \rangle > \pi \epsilon dn
		\end{equation}
and then use hyperplane rounding to find scalars $x_1,\ldots,x_n\in\{-1,1\}$ which satisfy:
		\begin{equation}\label{eqn:xixj} \sum_{i,j =1}^n M_{ij} x_ix_j>2\epsilon dn.\end{equation}
Let $z_1, \ldots, z_n\in\R^n$ be the columns of
\[
I + \frac{M}{\sqrt{d}},
\]
and note that $\|z_i\|^2\in [1,3]$ since $z_i(i)^2=1$ and every vertex is incident with at most $d$ edges in each of  $E(H)\setminus E(G)$ and $E(G)\setminus E(H)$. The relevant inner products of the $z_i$ are easy to understand:
			$$M_{ij}\langle z_i,z_j\rangle = \begin{cases} 0 & \textrm{if $M_{ij}=0$} \\
			\frac{2}{\sqrt{d}}=1\cdot(1\cdot \frac1{\sqrt{d}}+1\cdot \frac{1}{\sqrt{d}}) & \textrm{if }ij\in E(H)\setminus E(G)\\ \frac{2}{\sqrt{d}}=(-1)\cdot(1\cdot \frac{-1}{\sqrt{d}}+1\cdot \frac{-1}{\sqrt{d}}) & \textrm{if } ij\in E(G)\setminus E(H),
		\end{cases}$$
		where in the latter two cases we have used the fact that $i$ and $j$ cannot have any common neighbors because they lie on
		different sides of the bipartition.
Letting $y_i=z_i/\|z_i\|$, we therefore have
		\begin{equation}\label{eqn:2eps}
	\sum_{i,j = 1}^n M_{ij} \langle y_i, y_j \rangle \geq 2\cdot \delta dn\cdot \frac{2}{3\sqrt{d}}>\pi\epsilon dn
		\end{equation}
by our choice of $\delta$.
Let $w$ be a random unit vector and let 
		$$x_i = \begin{cases} +1 &\textrm{if }\langle y_i,w\rangle \ge 0\\
		-1 &\textrm{if }\langle y_i,w\rangle < 0.\end{cases}$$
We denote the angle between vectors $y_i, y_j \in Y$ as $\theta_{ij}$. We recall that $\theta_{ij} = \cos^{-1} \langle y_i, y_j \rangle$ since $y_i$ and $y_j$ have unit length.
It follows that the probability that $x_i \not = x_j$ is equal to $\frac{\theta_{ij}}{\pi}$ and the probability of $x_i = x_j$ is equal to $(1 - \frac{\theta_{ij}}{\pi})$.
Thus,
\begin{align*}
\E\left[\sum_{i,j = 1}^{n} M_{ij} x_i x_j\right] &= \sum_{i,j = 1}^{n} M_{ij} \E[x_i x_j] \\
&=  \sum_{i,j = 1}^{n}  M_{ij} \left(\frac{\theta_{ij}}{\pi} (-1) + (1- \frac{\theta_{ij}}{\pi})(1)\right) \\
&= \sum_{i,j = 1}^{n} M_{ij} \left(1- \frac{2\theta_{ij}}{\pi}\right) \\
&= \frac{2}{\pi}\sum_{i,j = 1}^{n} M_{ij} \left(\frac{\pi}{2} -  \theta_{ij}\right) \\
&= \frac{2}{\pi}\sum_{i,j = 1}^{n} M_{ij} \left(\frac{\pi}{2} -  \cos^{-1} \langle y_i, y_j \rangle\right) \\
&= \frac{2}{\pi}\sum_{i,j = 1}^{n} M_{ij} \sin^{-1} \langle y_i, y_j \rangle\\
	&\ge \frac{2}{\pi}\sum_{i,j = 1}^{n} M_{ij} \langle y_i, y_j \rangle\quad\textrm{since $\sin^{-1}(x)\ge x$}\\
	&>2\epsilon dn,
\end{align*}
		by \eqref{eqn:2eps}, as desired.

\end{proof}
\begin{remark} The analysis of the rounding scheme above can be improved from
$2/\pi$ to the Goemans-Williamson constant $0.868\ldots$, but we have chosen
not to do so for simplicity. 
\end{remark}

\begin{theorem}\label{thm:cutlb} For any $\epsilon=\omega(n^{-1/4})$, any $\epsilon$-cut sketching scheme $\sk$ for graphs with $n$ vertices
	must use at least $\frac{n\lg n}{2304\eps^2}\cdot (1-o_n(1))$ bits in the worst case.
\end{theorem}
\begin{proof} Assume $n$ 
	is divisible by $4$ (add a constant if this is not the case). Let $d=\frac{1}{16\cdot 9\epsilon^2}$ and let $B_{n,d}$ be the set of bipartite graphs on $n$ vertices with respect to a fixed bipartition. We proceed as in the proof of Theorem \ref{thm:speclb}. Let $f:\{-1,1\}^n\rightarrow \R$ be any function which is
	an $\epsilon-$cut sketch for some graph $H$. Arguing as in Lemma \ref{lem:sketchcount}, any other graph $G\in B_{n,d}$ which
	has the same sketch must $3\eps$-cut approximate $H$, so by Lemma \ref{lem:cutrigid}, any such $G$ must have at most
	$9\sqrt{d}\epsilon\cdot dn/2$ edges which are not present in $H$. Thus, the encoding length of such $H$ is at most
	$$ \ell:=dn/2 + 9d^{3/2}n\epsilon \lg n=dn\lg n/16\cdot(1+o(1))$$
	bits, by our choice of $d$, so any particular $f$ can only be an $\epsilon$-cut sketch for $2^\ell$ graphs in $B_{n,d}$

	On the other hand, $B_{n,d}$ is quite large.  Recall that the {\em
	bipartite double cover} of a graph $F$ on $n$ vertices is the graph on
	$2n$ vertices obtained by taking its tensor product with $K_2$, and two distinct
	graphs must have distinct double covers.  Thus, by \eqref{eqn:manygraphs} we have
	$$\lg |B_{n,d}|\ge \lg |G_{n/2,d}|\ge dn\lg n/8\cdot (1-o(1)).$$

	Thus, if $N$ is the number of distinct sketches produced by $\sk$, we must have
	$$(1+o(1))\lg N \ge dn\lg n/8 - dn\lg n/16 = dn\lg n/16 = n\lg n/(16)^2\cdot 9\epsilon^2,$$
	 as desired.
\end{proof}
\bibliography{main.bbl}
\end{document}